\documentclass[dvipdfmx, a4paper, 12pt]{amsart}
\usepackage[svgnames]{xcolor}
\usepackage{float}
\usepackage[pdftex]{graphicx}
\usepackage{tikz}
\usetikzlibrary{intersections, calc, arrows}
\usetikzlibrary{positioning}
\usetikzlibrary{chains}
\usepackage{epstopdf}
\allowdisplaybreaks[4]

\usepackage{amsmath,amssymb,amsthm}
\usepackage{amscd}
\usepackage{amsmath}
\usepackage{amssymb}
\usepackage{amsfonts}
\usepackage{cases}
\usepackage{color}
\usepackage{mathrsfs}
\usepackage{amsthm}
\usepackage{enumerate}
\usepackage[all]{xy} 
\usepackage{setspace}
\usepackage{latexsym}
\usepackage{wrapfig}
\usepackage{url}
\usepackage{braket}
\usepackage{array,enumerate}

\DeclareMathAlphabet{\mathpzc}{OT1}{pzc}{m}{it}
\theoremstyle{definition}
\newtheorem{defi}{Definition}[section]
\newtheorem{exam}[defi]{Example}

\theoremstyle{theorem}
\newtheorem{theo}[defi]{Theorem}
\newtheorem{prop}[defi]{Proposition}
\newtheorem{coro}[defi]{Corollary}

\theoremstyle{remark}
\newtheorem{rema}[defi]{Remark}

\makeatletter

\@addtoreset{equation}{section}

\title[Open quantum random walks on distance sets]{Hypergroup structures of open quantum random walks on distance sets}
\author[Y. Sawada]{Yusuke Sawada}
\address{Graduate school of mathematics, Nagoya University, Chikusaku, Nagoya, 464-8602, Japan}
\email{m14017c@math.nagoya-u.ac.jp}
\begin{document}
\maketitle
\begin{abstract}
Wildberger has introduced the method to construct a hermitian discrete hypergroup from a random walk on a graph. We will apply his method to an open quantum random walk (OQRW) on a distance set, and show that any discrete hypergroup which is not necessary hermitian is realized by an OQRW on a distance set. We will investigate distributions of OQRWs on distance sets in the view point of hypergroups.
\end{abstract}

\section{Introduction}

A discrete hypergroup is a probability theoretic extension of a discrete group and defined by a $*$-algebra whose involution corresponds to the inverse map. If the involution is trivial then the discrete hypergroup is said to be hermitian. A hermitian discrete hypergroup is the extension of a discrete group whose inverse map is the identity map. We also have the theory of general (continuous) hypergroups, however we will treat only discrete hypergroups in this paper.

In \cite{wild94} and \cite{wild95}, Wildberger has introduced the method to construct a hermitian (discrete) hypergroup from a random walk on a special pointed graph $(\Gamma,v_0)$. Here, a random walker leaves from the fixed vertex $v_0$ (called the base point), and jumps a distance which is not necessary $1$ in each step. If such a random walk on a pointed graph $(\Gamma,v_0)$ constructs a hermitian hypergroup in his method, then we simply say that the pointed graph $(\Gamma,v_0)$ produces the hermitian hypergroup. Any pointed graph does not always produce hermitian hypergroups, and it is known that any pointed distance regular graph produces a hermitian hypergroup whose structure is independent of the choice of the base point (see \cite{ikka-sawa19}). We can apply any pointed graph to Wildberger's method and obtain a hermitian hypergroup, whose product may fail the accosiativity, called a pre-hypergroup.

In \cite{endo-mimu-sawa20}, when a pointed graph $(\Gamma,v_0)$ satisfies a symmetry condition $(S)$ which is weaker than the distance regularity, it has been shown that a distance distribution of random walks on $\Gamma$ can be computed by non-associative algebraic structure of the pre-hypergroup derived from $\Gamma$. We can regard that the concept of random walk on a pointed graph with the condition $(S)$ is included in the one of pre-hypergroup in the view point of distance distribution, as Figure 1.1. Any example of a pointed graph satisfying the condition $(S)$ and not producing a hermitian hypergroup, has not been found yet.

The concept of open quantum random walk (OQRW) has been introduced by Attal-Petruccione-Sabot-Sinayskiy in \cite{aps12} and \cite{apss12} as quantum Markov chains on graphs based on the non-unitary dynamics. Any classical Markov chain is recovered by OQRW. In \cite{agps15}, Attal-Guillotin-Plantard-Sabot have established the central limit theorem for OQRW. In \cite{dhah-ko-yoo19}, an OQRW is associated with a quantum Markov chain in the sense of \cite{acca-koro91}, and investigated the irreducibility and reducibility.

We shall give an outline of this paper as follows:

In Section 2, we will recall the concept of discrete hypergroup, Wildberger's construction, and the relationship between a distance distribution of a random walk on a pointed graph $\Gamma$ equipped with the symmetry condition $(S)$ and the structures of the pre-hypergroup derived from $\Gamma$.

In Section 3, influenced by OQRW in \cite{aps12} and \cite{apss12}, we will introduce a concept of open quantum random walk (OQRW) on distance set as a quantum analogy of a time evolution of a distance between the base point and a random vertex given by a random walk on a pointed graph. A distance set is $\{0,1,\ldots,N\}$ for some $N\in\mathbb{N}$ or $\mathbb{N}\cup\{0\}$. An OQRW $\mathcal{M}$ on distance set $\mathcal{D}$ will be defined as a family of maps, which acts on positive trace class operators, parametrized by elements in $\mathcal{D}$, and each element in $\mathcal{M}$ can be regarded as an OQRW in the sense of \cite{apss12}.

In Section 4, we will apply the idea of Wildberger's construction to an OQRW $\mathcal{M}$ on a distance set and give a method to construct a hypergroup (which is not necessarily hermitian) from $\mathcal{M}$. In this construction too, any OQRW on a distance set does not always produce a hypergroup. An arbitrary hypergroup $H$ will be realized by an OQRW $\mathcal{M}_H$ on a distance set, like the recovery of Markov chains by OQRWs. We can also obtain OQRW on a distance set from a pre-hypergroup by the method.

In Section 5, it will be shown that if an OQRW $\mathcal{M}$ on a distance set has a hypergroup-like structure, then each distribution of $\mathcal{M}$ described by the structure. This is a quantum analogy of the result in \cite{endo-mimu-sawa20}, and implies that a distribution of the OQRW $\mathcal{M}_H$ on the distance set derived from a hypergroup $H$ can be computed by the hypergroup structure. In this sense, we can regard that the concept of hypergroup is included in the one of OQRW on distance set as Figure 1.1. The hypergroup-like structure of $\mathcal{M}_H$ corresponds to the associativity of $H$. Thus, we can not obtain the above result for the OQRW on the distance set derived from a pre-hypergroup which is not a hypergroup.

In the future, it is hoped to investigate the limit theorem for OQRWs on distance sets and some connections with quantum information theory in \cite{kemp03}.

At the end of this section, we prepare some notations used in this paper as follows: let $\mathbb{N}$ and $\mathbb{C}$ denote the set of all integers greater than $0$ and the set of all complex numbers, respectively. Also, we put $\mathbb{N}_0=\mathbb{N}\cup\{0\}$.

For a set $S$, we denote by $\mathbb{C}S$ the free vector space of $S$ over $\mathbb{C}$ and by $|S|$ the cardinality of $S$.

Let $\mathcal{H}$ be a Hilbert space. We denote by $\mathcal{B}(\mathcal{H})$ the algebra consisting of all bounded linear operators on $\mathcal{H}$, and by $\mathcal{B}_1(\mathcal{H})$ the Banach space consisting of all trace class operators on $\mathcal{H}$ with the trace norm $\|\cdot\|_1$. Also, we define
\[
\mathcal{B}_{1}(\mathcal{H})_{+,1}=\{X\in\mathcal{B}_1(\mathcal{H})\mid
X\geq0,\ {\rm
Tr}(X)=1\},
\]
where ${\rm
Tr}(X)$ is the canonical trace of $X$. The identity operator and the zero map on $\mathcal{H}$ are denoted by ${\bf
1}$ and $O$, respectively. For two vectors $\xi,\eta\in\mathcal{H}$, the operator $\ket{\xi}\!\bra{\eta}$ is the one rank operator defined by
\[
\ket{\xi}\!\bra{\eta}(\zeta)=\langle\eta,\zeta\rangle\xi
\]
for each $\zeta\in\mathcal{H}$.\\

\noindent\begin{tikzpicture}
\draw (6,9.5) -- (11,9.5)  --(11,11) -- (6,11)  --  (6,9.5);
\draw (5,9.5) -- (12,9.5)  --(12,12.5) -- (5,12.5)  --  (5,9.5);
\draw[very thick]  (4,9.5) -- (11,9.5)  --(11,14) -- (4,14)  --  (4,9.5);
\filldraw[fill=blue, opacity=.1, draw=Indigo]  (4,9.5) -- (14.5,9.5)  --(14.5,14) -- (4,14)  --  (4,9.5);
\draw (3,9.5) -- (11,9.5)  --(11,15) -- (3,15)  --  (3,9.5);
\draw (2,9.5) -- (11,9.5)  --(11,16) -- (2,16)  --  (2,9.5);
\draw (12.7,14.3)node{pre-hypergroups};
\draw (7.5,14.3)node{hermitian hypergroups};
\draw (8,13.3)node{random walks on};
\draw (8,12.8)node{pointed graphs with $(S)$};
\draw (8.5,11.8)node{random walks on pointed};
\draw (8.5,11.3)node{distance regular graphs};
\draw (7,15.3)node{hypergroups};
\draw (6.5,16.3)node{open quantum random walks on distance sets};
\draw (11.5,11)node{?};
\draw (8,8.5)node{Figure 1.1};
\end{tikzpicture}

\section{Hypergroups and Wildberger's construction}\label{HGwildberger}
In this section, we will recall Wildberger's method to construct a hermitian hypergroup from a random walk on a pointed graph, and a description of a distance distribution of the random walk by the associated (pre-)hypergroup structure. 

The concept of hypergroup is the probability theoretic extension of the one of topological group, and introduced by Dunkl \cite{dunk73}, Jewett \cite{jewe75} and Spector \cite{spec78}. In this paper, we only treat discrete hypergroups. We refer the reader to \cite{lass05} for the general theory of discrete hypergroups (also, see \cite{bloo-heye95} for the theory of general hypergroups). We give the definition of discrete hypergroup as follows:
\begin{defi}
Suppose $H=\{x_i\}_{i\in
I(H)}$ is a countable set whose elements are parametrized by elements in $I(H)$, where $I(H)$ is $\{0,\ldots,N\}$ for some $N\in\mathbb{N}$ or $\mathbb{N}_0=\mathbb{N}\cup\{0\}$. If a binary operation $\circ$ and a map $*$ on the free vector space $\mathbb{C}H$ satisfy the following properties then the triple $(H,\circ,*)$ is called a discrete hypergroup.
\begin{enumerate}
\item
The restriction $*|_H$ maps onto $H$ and the triple $(\mathbb{C}H,\circ,*)$ is a $*$-algebra with the unit $x_0\in
H$.
\item
For each $i,j\in I(H)$, there exist $m\in
I(H)$ and $q_{i,j}^0,\ldots,q_{i,j}^m\geq0$ with $\sum_{k=0}^mq_{i,j}^k=1$ such that
\[
x_i\circ x_j=\sum_{k=0}^mq_{i,j}^kx_k.
\]
\item
For all $i,j\in I(H)$, we have $q_{i,j}^0\neq0$ if and only if $x_i=x_j^*$.
\end{enumerate}
A hypergroup $(H,\circ,*)$ is said to be hermitian if the restriction $*|_H$ is the identity map on $H$. A discrete pre-hypergroup is a hermitian discrete hypergroup $(H,\circ,*)$ such that the algebra $(\mathbb{C}H,\circ)$ may fail the associativity.

We call the numbers $\{q_{i,j}^k\}_{i,j,k\in I(H)}$ the structure constants of a hypergroup (or a pre-hypergroup) $(H,\circ,*)$.

In this paper, a discrete hypergroup (or a discrete pre-hypergroup) $(H,\circ,*)$ is simply called a hypergroup (or a pre-hypergroup, respectively) and denoted by $H$.
\end{defi}

Suppose $H=\{x_i\}_{i\in
I(H)}$ is a hypergroup (or a pre-hypergroup) and $\{q_{i,j}^k\}_{i,j,k\in
I(H)}$ are the structure constants of $H$. For a tuple $(k_1,\ldots,k_n)\in
I(H)^n$ with $n\geq2$, we have
\[
(( \cdots ((x_{k_1} \circ x_{k_2}) \circ x_{k_3}) \circ \cdots ) \circ x_{k_{n-1}})\circ  x_{k_n}=\sum_{m\in
I(H)}q_{k_1,\cdots,k_n}^mx_m
\]
for some $q_{k_1,\cdots,k_n}^m\geq0$ with $\sum_{m\in
I(H)}q_{k_1,\cdots,k_n}^m=1.$ Here, the left hand side of the above equation means the $n-1$ times products from the left step by step when $H$ is a pre-hypergroup. The numbers
\[
\{q_{k_1,\ldots,k_n}^m\mid
m\in\mathbb{N}_0,n\geq2,k_1,\ldots,k_n\in
I(H)\}
\]
is called the multi-structure constants of $H$.
\begin{rema}
If a hypergroup $(H,\circ,*)$ is hermitian then the algebra $(\mathbb{C}H,\circ)$ is commutative.
\end{rema}
The concept of isomorphism between hypergroups are defined as follows:

\begin{defi}
For two hypergroups $H$ and $H'$, a map $\Phi:H\to
H'$ is called a homomorphism from $H$ to $H'$ if $\Phi$ can be extended to a $*$-homomorphism from $\mathbb{C}H$ to $\mathbb{C}H'$. When $\Phi$ is bijective it is called an isomorphism from $H$ onto $H'$.
\end{defi}

We provide some notations and basic assumption with respect to graphs as the following. A pointed graph is a pair $(\Gamma,v_0)$ of a graph $\Gamma$ with a vertex set $V$ and a vertex $v_0\in
V$ called a base point. For each $v,w\in
V$, we denote by $d(v,w)$ the length of the shortest path from $v$ to $w$, called the distance between $v$ and $w$. We define a set $I(\Gamma,v_0)$ of distances between vertices by
\[
I(\Gamma,v_0)=\{n\in\mathbb{N}_0\mid
d(v,v_0)=n\ \mbox{for some }v\in
V\},
\]
and a sphere $S_n(v)$ of radius $n\in
I(\Gamma,v_0)$ centered at $v\in
V$ by
\[
S_{n}(v)=\{w\in
V\mid
d(v,w)=n\}.
\]
In this paper, assume that any graph $\Gamma$ is simple, connected, locally finite and has at most countable vertices, and any pointed graph $(\Gamma,v_0)$ satisfies $S_n(v) \neq \varnothing$ for all vertex $v$ and $n\in
I(\Gamma,v_0)$.

Now, we shall recall the construction of a hermitian hypergroup from a random walk on a pointed graph, introduced by Wildberger in \cite{wild94} and \cite{wild95} (see also \cite{ikka-sawa19}). Here, a random walk on a pointed graph $(\Gamma,v_0)$ means a random walk on $\Gamma$ leaving from the vertex $v_0$, that a random walker jumps a distance which is not necessary $1$ in each step.

Let $(\Gamma, v_0)$ be a pointed graph and put $H(\Gamma,v_0)=\{x_i\}_{i\in I(\Gamma,v_0)}$ with dummy symbols $x_i$ for each $i\in
I(\Gamma,v_0)$. We define a product on the free vector space $\mathbb{C}H(\Gamma,v_0)$ by
\[
x_i\circ
x_j=\sum_{k\in
I(\Gamma,v_0)}p_{i,j}^kx_k
\]
for each $i,j\in
I(\Gamma,v_0)$, where
\begin{equation}\label{wildp}
p_{i,j}^k=\frac{1}{|S_{i} (v_0)|} \sum_{v \in S_{i} (v_0)}\frac{|S_{j}(v)\cap S_{k} (v_0)|}{|S_{j} (v)|}.
\end{equation}
The sequance $(p_{i,j}^k)_{k\in
I(\Gamma,v_0)}$ is the distribution of distances between the base point $v_0$ and a random vertex $w\in S_{j} (v)$ after a jump to a random vertex $v\in S_{i}(v_0)$. We also define an involution on $\mathbb{C}H(\Gamma,v_0)$ by $x_i^*=x_i$ for each $i\in
I(\Gamma,v_0)$. Then $H(\Gamma,v_0)$ forms a pre-hypergroup with the structure constants $\{p_{i,j}^k\}_{i,jk\in
I(\Gamma,v_0)}$. If $H(\Gamma,v_0)$ forms a hermitian hypergroup in the above construction, then we say that the pointed graph $(\Gamma,v_0)$ produces the hermitian hypergroup $H(\Gamma,v_0)$.

All pointed graphs do not always produce hermitian hypergroups. By \cite[Proposition 3.2]{ikka-sawa19}, to check that $H(\Gamma,v_0)$ forms a hermitian hypergroup, it is enough to show that $\mathbb{C}H(\Gamma,v_0)$ satisfies the commutativity and the associativity. A distance regular graph is a graph equipped with a good symmetry with respect to distances, see \cite{bcn89} for the definition. It is known that any distance regular graph produces a hermitian hypergroup whose structure is independent of the choice of a base point, see \cite[Theorem 3.3]{ikka-sawa19} for the proof.
\begin{exam}\label{cycle4}
We denote by $\mathcal{C}_4$ the Cayley graph ${\rm
Cay}(\mathbb{Z}/4\mathbb{Z},\{\overline{1},\overline{3}\})$ drawn as Figure 2.1. We have $I(\mathcal{C}_4,\overline{0})=\{0,1,2\}$ and the pointed graph $(\mathcal{C}_4,\overline{0})$ produces a hermitian hypergroup $H({\rm
Cay}(\mathbb{Z}/4\mathbb{Z},\{\overline{1},\overline{3}\}),\overline{0})$ with the following structure:
\[
x_1\circ
x_1=\frac{1}{2}x_0+\frac{1}{2}x_2;\quad
x_1\circ
x_2=x_2\circ
x_1=x_1;\quad
x_2\circ
x_2=x_0.
\]
\[
\begin{xy}
(0,0)*{\circ}="A",
(-15,0)*{\circ}="B",
(0,-15)*{\circ}="C",
(-15,-15)*{\circ}="D",
{ "A" \ar @{-} "B" },
{ "A" \ar @{-} "C" },
{ "C" \ar @{-} "D" },
{ "B" \ar @{-} "D" },
(-17,4)*{\overline{0}},
(-7.5,-22)*{\mbox{{\rm
Figure }}2.1},
\end{xy}
\]

\end{exam}

\begin{exam}\label{lattice}
We consider the Cayley graph ${\rm
Cay}(\mathbb{Z},\{\pm1\})$ drawn as Figure 2.2, with a base point $0$. Then we have $I({\rm
Cay}(\mathbb{Z},\{\pm1\}),0)=\mathbb{N}_0$ and the pointed graph $({\rm
Cay}(\mathbb{Z},\{\pm1\}),0)$ produces the hermitian hypergroup $H({\rm
Cay}(\mathbb{Z},\{\pm1\}),0)$ with the structure given by
\[
x_i\circ
x_j=\frac{1}{2}x_{|i-j|}+\frac{1}{2}x_{i+j}
\]
for each $i,j\in\mathbb{N}_0$.
\[
\begin{xy}
(0,0)*{\circ}="A",
(-10,0)*{\circ}="B",
(-20,0)*{\circ}="C",
(-30,0)*{\circ}="D",
(10,0)*{\circ}="E",
(20,0)*{\circ}="F",
(30,0)*{\circ}="G",
{ "A" \ar @{-} "B" },
{ "B" \ar @{-} "C" },
{ "C" \ar @{-} "D" },
{ "A" \ar @{-} "E" },
{ "E" \ar @{-} "F" },
{ "F" \ar @{-} "G" },
(0,5)*{0},
(0,-7)*{\mbox{{\rm
Figure }}2.2\label{lattice}},
\end{xy}
\]

In general, let $F_n$ be the $n$-free group with the generator $A =\{ a_1, a_2, \ldots, a_n \}$.
The Cayley graph ${\rm Cay}(F_n,A\cup A^{-1})$ is distance regular, and hence it produces a hermitian hypergroup (see \cite[Corollary 3.8]{ikka-sawa19} for the structure constants).
\end{exam}
We can find many examples of non-distance regular graphs producing hypergroups in \cite[Section 4]{ikka-sawa19}.

We need only the first $2$ steps of a random walk on a pointed graph $(\Gamma,v_0)$ in Wildberger's construction. However, a distance distribution given by steps of an arbitrary number of times is described by a product on the pre-hypergroup $H(\Gamma,v_0)$ if the pointed graph $(\Gamma,v_0)$ satisfies a symmetry condition as follows:
\begin{theo}[{\rm\cite[Theorem 4.5]{endo-mimu-sawa20}}]\label{ems}
Let $H(\Gamma,v_0)=\{x_i\}_{i\in
I(\Gamma,v_0)}$ be the pre-hypergroup derived from a pointed graph $(\Gamma,v_0)$ satisfying the following symmetry condition:
\begin{itemize}
\item[$(S)$] the function $|S_i(\cdot)|$ is a constant on the vertex set, and the function $| S_i(\cdot) \cap S_j(v_0) |$ is a constant on $S_k( v_0 )$ for each $i, j, k \in I(\Gamma,v_0)$.
\end{itemize}
For each $k_1, \ldots , k_n \in I(\Gamma,v_0)$, we denote by $p_{k_1,\ldots,k_n}^m$ the coefficient of $x_m$ in the following element:
\begin{equation}\label{probp}
\sum_{v_1 \in S_{k_1}(v_0)}\sum_{v_2 \in S_{k_2}(v_1)} \cdots \sum_{v_n \in S_{k_n}(v_{n -1})}
\frac{1}{\prod_{j =1}^{n}| S_{k_j} ( v_{j -1} ) |} x_{d(v_n,v_0)}.
\end{equation}
Then $p_{k_1,\ldots,k_n}^m$ coincides with the multi-structure constant $q_{k_1,\ldots,k_n}^m$ of $H(\Gamma,v_0)$.
\end{theo}


The sequence $(p_{k_1,\ldots,k_n}^m)_{m\in
I(\Gamma,v_0)}$ is the conditional probability that the distance between $v_0$ and a vertex which the random walker reaches is $m$, under the $n$-times jumps $v_0\xrightarrow{k_1}\cdot\xrightarrow{k_2}\cdots\xrightarrow{k_n}\cdot\ $. The previous theorem enable us to compute a distribution by the (non-associative) algebraic structure.

\begin{rema}

\ 

\begin{enumerate}
\item
In the case when $n=2,i=k_1,j=k_2,k=m$, the coefficient of $x_m$ in \eqref{probp} coincides with \eqref{wildp}, and hence the number $p_{k_1,k_2}^m$ is well-defined.
\item
The condition $(S)$ is weaker than the distance regularity, and it is unclear whether any pointed graph with $(S)$ produces a hermitian hypergroup. We refer the reader to \cite[Example 4.6]{endo-mimu-sawa20} for an example of a pointed graph $(\Gamma,v_0)$ with $(S)$ such that $\Gamma$ is not distance regular.
\end{enumerate}
\end{rema}

Suppose a pointed graph $(\Gamma,v_0)$ produces a hermitian hypergroup $H(\Gamma,v_0)$ and $\{p_{i,j}^k\}_{i,j,k\in I(\Gamma,v_0)}$ is the structure constants of $H(\Gamma,v_0)$. A family $\mathcal{P}_{H(\Gamma,v_0)}=\{P_k\}_{k\in I(\Gamma,v_0)}$ of transition matrices $P_k = (p_{k,i}^j)_{i,j \in I (\Gamma,v_0)}$ satisfies that
\begin{equation}\label{matrixproduct}
P_iP_j=\sum_{k\in
I(\Gamma,v_0)}p_{i,j}^kP_k,
\end{equation}
and hence $P_i$ and $P_j$ commute, for all $i,j \in I(\Gamma,v_0)$. Also, the family $\mathcal{P}_{H(\Gamma,v_0)}$ is linear independent. In other words, we have an algebraic isomorphism between $\mathbb{C}H(\Gamma,v_0)$ and the commutative matrix algebra with the basis $\mathcal{P}_{H(\Gamma,v_0)}$. Each matrix $P_k$ can be regarded as bounded operators on the Hilbert space $\ell^2(I(\Gamma,v_0))$ consisting of all square summable sequences on $I(\Gamma,v_0)$ by
\[
(P_k(\xi))_n
=\sum_{l\in
I(\Gamma,v_0)}p_{k,l}^n\xi_l
\]
for each $\xi\in\ell^2(I(\Gamma,v_0))$, see \cite[Theorem 5.3]{endo-mimu-sawa20} for the norm estimation. The action can be regarded as to multiply the matrix $P_k$ by a row vector in $\ell^2(I(\Gamma,v_0))$ from the left. By Theorem \ref{ems} and \eqref{matrixproduct}, we have the following corollary.
\begin{coro}{\rm{(\cite[Corollary 5.7]{endo-mimu-sawa20})}}\label{maincoro}
Suppose a pointed graph $(\Gamma,v_0)$ satisfies the condition $(S)$ and produces a hermitian hypergroup $H(\Gamma,v_0)$. Let $\mathcal{P}_{H(\Gamma,v_0)}$ be the family of the transition matrices associated with $H(\Gamma,v_0)$. We have
\[
P_{k_n}\cdots
P_{k_1}=\sum_{m\in
I(\Gamma,v_0)}p_{k_n,\ldots,k_1}^mP_m
\]
for all $m\in\mathbb{N}$ and $k_1,\ldots,k_n\in
I(\Gamma,v_0)$.
\end{coro}
The above corollary implies that the vector $(P_{k_n}\cdots
P_{k_1})p^{(0)}$ coincides with the distribution $(p_{k_1,\ldots,k_n}^i)_{i\in
I(\Gamma,v_0)}$ for the initial distribution $p^{(0)}=(1,0,0,\ldots)\in\ell^2(I(\Gamma,v_0))$, like computations of distributions by a transition matrix in a Markov chain.

By the commutativity of the hypergroup $H(\Gamma,v_0)$, we have
\[
p_{k_n,\ldots,k_1}^m=p_{\sigma(k_n),\ldots,\sigma(k_1)}^m
\]
for every permutation $\sigma$ on $\{k_1,k_2,\ldots,k_n\}$.

\section{open quantum random walks on distance sets}
An open quantum random walk (OQRW) has been introduced in \cite{aps12} and \cite{apss12}. In this section, we will introduce a concept of open quantum random walk on distance set as a family of OQRW. This is a quantum analogy of a time evolution of a distance between the base point and a random vertex given by a random walk on a pointed graph in the previous section.

Suppose $\mathcal{D}$ is the set $\{0,\ldots,N\}$ for some $N\in\mathbb{N}$ or $\mathbb{N}_0$. The term ``distance set'' usually means a family of distances between pairs of points in a given point set, however in this paper we call $\mathcal{D}$ a distance set too. For a hypergroup $H$ and a pointed graph $(\Gamma,v_0)$, the sets $I(H)$ and $I(\Gamma,v_0)$ are called the distance sets associated with $H$ and $(\Gamma,v_0)$, respectively.

Let $\mathcal{D}$ be a distance set, $\mathcal{K}$ a Hilbert space with an orthonormal basis $\{\ket{n}\}_{n\in\mathcal{D}}$ parametrized by elements in $\mathcal{D}$, and $\mathcal{H}$ a separable Hilbert space. Suppose $\{B_{i,j;k}\}_{i,j,k\in\mathcal{D}}$ is a family of bounded operators on $\mathcal{H}$ satisfying
\begin{equation}\label{OQRWB}
\sum_{i\in\mathcal{D}}B_{i,j;k}^*B_{i,j;k}={\bf1}
\end{equation}
for each $j,k\in\mathcal{D}$. Here, the above convergence is with respect to the strong operator topology in $\mathcal{B}(\mathcal{H})$ when $\mathcal{D}=\mathbb{N}_0$. The Hilbert space $\mathcal{K}$ is the space of distances (or positions in a broad sense), the Hilbert space $\mathcal{H}$ means the space of degrees of freedom, and the operator $B_{i,j;k}$ describes the effect of passing from $j$ to $i$ by a jump with distance $k$. For $i,j,k\in\mathcal{D}$, we define an operator
\[
M_{i,j;k}=B_{i,j;k}\otimes\ket{i}\!\bra{j}
\]
on $\mathcal{H}\otimes\mathcal{K}$. By \cite[Lemma 2.1, 2.2]{apss12}, it turns out that a series
\[
\mathcal{M}_k(X)=\sum_{i,j\in\mathcal{D}}M_{i,j;k}XM_{i,j;k}^*
\]
converges to a positive trace class operator with respect to the trace norm for each positive operator $X\in\mathcal{B}_1(\mathcal{H}\otimes\mathcal{K})$ and $k\in\mathcal{D}$. The map $\mathcal{M}_k$ can be extended to a map on $\mathcal{B}_1(\mathcal{H}\otimes\mathcal{K})$. The extension is also denoted by $\mathcal{M}_k$.
\begin{defi}
The family $\mathcal{M}=\{\mathcal{M}_k\}_{k\in\mathcal{D}}$ of the maps $\mathcal{M}_k$ on $\mathcal{B}_1(\mathcal{H}\otimes\mathcal{K})$ is called an open quantum random walk (OQRW) on a distance set $\mathcal{D}$.
\end{defi}

A quantum walker jumps between vertices in a graph in OQRWs in \cite{aps12} and \cite{apss12}, while a quantum walker jumps between natural numbers whose distance is not necessary $1$ in our OQRWs on distance sets. We can regard an OQRW $\mathcal{M}=\{\mathcal{M}_k\}_{k\in\mathcal{D}}$ on a distance set $\mathcal{D}$ as a family of OQRWs in \cite{aps12} and \cite{apss12}. Thus, it is often possible to apply each component $\mathcal{M}_k$ to an existing result on OQRWs as the following proposition:
\begin{prop}\label{oqwlemm}{\rm
(}\cite[Lemma 2.4, Proposition 2.3]{apss12}{\rm
)}
We define the isometry $U_j:\mathcal{H}\ni\xi\mapsto\xi\otimes\ket{j}\in\mathcal{H}\otimes\mathcal{K}$ for each $j\in\mathcal{D}$. For $\rho\in\mathcal{B}_1(\mathcal{H}\otimes\mathcal{K})_{+,1}$, the positive operator $\rho_j=U_j^*\rho
U_j\in\mathcal{B}_1(\mathcal{H})$ satisfies
\[
({\bf
1}\otimes\ket{i}\!\bra{j})\rho({\bf
1}\otimes\ket{j}\!\bra{i})=\rho_j\otimes\ket{i}\!\bra{i},\quad{\rm
Tr}(\rho_j)={\rm
Tr}(\rho({\bf
1}\otimes\ket{j}\!\bra{j})).
\]
If $\mathcal{M}=\{\mathcal{M}_k\}_{k\in\mathcal{D}}$ is an OQRW on $\mathcal{D}$ then each series $\sum_{j\in
\mathcal{D}}B_{i,j;k}\rho_jB_{i,j;k}^*$ converges to a positive trace class operator with respect to the trace norm on $\mathcal{B}_1(\mathcal{H})$, and satisfies
\[
\sum_{i\in\mathcal{D}}{\rm
Tr}\left(\sum_{j\in\mathcal{D}}B_{i,j;k}\rho_jB_{i,j;k}^*\right)=1,\quad\mathcal{M}_{k}(\rho)=\sum_{i\in
\mathcal{D}}\sum_{j\in
\mathcal{D}}B_{i,j;k}\rho_jB_{i,j;k}^*\otimes\ket{i}\!\bra{i}
\]
for each $k\in\mathcal{D}$, where the second convergence is in the trace norm on $\mathcal{B}_1(\mathcal{H}\otimes\mathcal{K})$. Thus, we have $\mathcal{M}_k(X)\in\mathcal{B}_1(\mathcal{H}\otimes\mathcal{K})_{+,1}$ for each $X\in\mathcal{B}_1(\mathcal{H}\otimes\mathcal{K})_{+,1}$.
\end{prop}
In \cite{apss12}, the operator $\rho_j$ in the above proposition is denoted by $\bra{j}\rho\ket{j}$.

If a state $\rho\in\mathcal{B}_1(\mathcal{H}\otimes\mathcal{K})_{+,1}$ has the form
\[
\rho=\sum_{i\in\mathcal{D}}\rho_i\otimes\ket{i}\!\bra{i}
\]
for some positive operators $\rho_i\in\mathcal{B}_1(\mathcal{H})$ with $\sum_{i\in\mathcal{D}}{\rm
Tr}(\rho_i)=1$ then a measurement along the basis $\{\ket{n}\}_{n\in\mathcal{D}}$ gives $\ket{i}$ with the probability ${\rm
Tr}(\rho_i)$. We define a distribution $d(\rho)$ of $\rho$ by
\[
d(\rho)=({\rm
Tr}(\rho_i))_{i\in\mathcal{D}}.
\]
For an OQRW $\mathcal{M}=\{\mathcal{M}_k\}_{k\in\mathcal{D}}$ on a distance set $\mathcal{D}$, we have
\[
\mathcal{M}_k(\rho)=\sum_{i,j\in\mathcal{D}}B_{i,j;k}\rho_jB_{i,j;k}^*\otimes\ket{i}\!\bra{i}
\]
for all $k\in\mathcal{D}$ by \cite[Corollary 2.5]{apss12}.

We will consider a dynamical system given by $\mathcal{M}$ with an initial state $\rho^{(0)}=\sum_{i\in\mathcal{D}}\rho_i^{(0)}\otimes\ket{i}\!\bra{i}\in\mathcal{B}_1(\mathcal{H}\otimes\mathcal{K})_{+,1}$. For a tuple ${\bf
k}=(k_1,\ldots,k_n)\in\mathcal{D}^n$, we define a state $\rho^{(n;{\bf
k})}$ and operators $\rho_i^{(n;{\bf
k})}$ by
\[
\rho^{(n;{\bf
k})}=\mathcal{M}_{k_n}\circ\cdots\circ\mathcal{M}_{k_1}(\rho^{(0)})=\sum_{i\in\mathcal{D}}\rho_i^{(n;{\bf
k})}\otimes\ket{i}\!\bra{i},\]
and write $p^{(0)}=d(\rho^{(0)})$ and $p^{(n;{\bf
k})}=d(\rho^{(n;{\bf
k})})$. The probability which a measurement gives $\ket{i}$ is $p_i^{(n;{\bf
k})}$ after $n$ jumps $(k_1,\ldots,k_n)$ (cf. the distribution $(p_{k_1,\ldots,k_n}^m)_{m\in
I(\Gamma,v_0)}$ defined in Theorem \ref{ems}).

Recall that the family $\mathcal{P}_{H(\Gamma,v_0)}$ is linear independent when a pointed graph $(\Gamma,v_0)$ produces a hypergroup $H(\Gamma,v_0)$. We shall investigate the linear independence of an OQRW on a distance set, on the Banach space $\mathcal{L}(\mathcal{B}_{1}(\mathcal{H}\otimes\mathcal{K}))$ consisting of all bounded maps on $\mathcal{B}_{1}(\mathcal{H\otimes\mathcal{K}})$ as follows:

\begin{prop}\label{inde}
If an open quantum random walk $\mathcal{M}=\{\mathcal{M}_k\}_{k\in\mathcal{D}}$ on a distance set $\mathcal{D}$ satisfies either of the following two conditions then the family $\{\mathcal{M}_k\}_{k\in\mathcal{D}}$ is linear independent on $\mathcal{L}(\mathcal{B}_{1}(\mathcal{H}\otimes\mathcal{K}))$.
\begin{enumerate}[{\rm
(1)}]
\item
There exist $j_0\in\mathcal{D}$ and $\xi_0\in\mathcal{H}$ such that the family $\{B_{i,j_0;k}\xi_0\}_{k\in\mathcal{D}}$ is linear independent on $\mathcal{H}$ for all $i\in\mathcal{D}$.
\item
There exist $j_0\in\mathcal{D}$ and a family $\{U_k\}_{k\in\mathcal{D}}$ of isometry operators $U_k$ on $\mathcal{H}$ such that
\[
B_{i,j_0,k}=\begin{cases}
U_k\quad{(i=k)}\\
O\quad{(i\neq
k)}.
\end{cases}
\] 
\end{enumerate}
\end{prop}
\begin{proof}
Let $\{\alpha_k\}_{k\in\mathcal{D}_0}$ be a family of complex numbers satisfies that $\sum_{k\in\mathcal{D}_0}\alpha_k\mathcal{M}_k=O$ on $\mathcal{B}_{1}(\mathcal{H}\otimes\mathcal{K})$, where $\mathcal{D}_0$ is an arbitrary finite subset of $\mathcal{D}$.\

Assume the condition $(1)$. We have
\begin{align*}
&0=\left(\sum_{k\in\mathcal{D}_0}\alpha_k\mathcal{M}_k(\ket{{\xi}_0}\!\bra{{\xi}_0}\otimes\ket{j_0}\!\bra{j_0})\right)\left(\sum_{l\in\mathcal{D}_0}B_{l,j_0;l}\xi_0\otimes\ket{l}\right)\\
&=\sum_{i\in\mathcal{D}_0}\sum_{k\in\mathcal{D}_0}\alpha_kB_{i,j_0;k}\ket{{\xi}_0}\!\bra{{\xi}_0}B_{i,j_0;k}^*B_{i,j_0;i}\xi_0\otimes\ket{i}\\
&=\sum_{i\in\mathcal{D}_0}\sum_{k\in\mathcal{D}_0}\alpha_k\langle\xi_0,B_{i,j_0;k}^*B_{i,j_0;i}\xi_0\rangle\
B_{i,j_0;k}\xi_0\otimes\ket{i}
\end{align*}
Fix an arbitrary $i\in\mathcal{D}_0$. By the above equation, we have
\[
\sum_{k\in\mathcal{D}_0}\alpha_k\langle\xi_0,B_{i,j_0;k}^*B_{i,j_0;i}\xi_0\rangle\
B_{i,j_0;k}\xi_0=0.
\]
The linear independence of $\{B_{i,j_0;k}\xi_0\}_{k\in\mathcal{D}}$ implies that $B_{i,j_0;i}\xi_0\neq0$ and $\alpha_k\langle\xi_0,B_{i,j_0;k}^*B_{i,j_0;i}\xi_0\rangle=0$ for all $k\in\mathcal{D}_0$, and hence we have $\alpha_i\langle\xi_0,B_{i,j_0;i}^*B_{i,j_0;i}\xi_0\rangle=0$ equivalently $\alpha_i=0$.

Assume the condition $(2)$. For an arbitrary operator $\rho\in\mathcal{B}_1(\mathcal{H})_{+,1}$, we have
\begin{align*}
&0=\sum_{k\in\mathcal{D}_0}\alpha_k\mathcal{M}_k(\rho\otimes\ket{j_0}\!\bra{j_0})=\sum_{k\in\mathcal{D}_0}\alpha_kU_k\rho
U_k^*\otimes\ket{k}\!\bra{k},
\end{align*}
and hence each trace of $\alpha_kU_k\rho
U_k^*$ is $0$. This implies that $\alpha_k=0$ for all $k\in\mathcal{D}_0$.
\end{proof}


\section{Realization of hypergroups}
In \cite[Section 6]{apss12}, any Markov chain is recovered by an open quantum random walk. In this section, we shall apply Wildberger's method in Section \ref{HGwildberger} to an OQRW on a distance set. Any not necessarily hermitian hypergroup will be realized by an OQRW on a distance set in a method influenced by \cite[Section 6]{apss12}. 

Let $\mathcal{M}=\{\mathcal{M}_k\}_{k\in\mathcal{D}}$ be an OQRW on a distance set $\mathcal{D}$ and $\rho^{(0)}=\sum_{i\in\mathcal{D}}\rho_i^{(0)}\otimes\ket{i}\!\bra{i}\in\mathcal{B}_1(\mathcal{H}\otimes\mathcal{K})_{+,1}$ an initial state. We define a non-negative number
\[
\tilde{p}_{k,l}^m={\rm
Tr}(\rho_m^{(2;(l,k))})
\]
for each $k,l,m\in\mathcal{D}$. It is clear that $(\tilde{p}_{k,l}^m)_{m\in\mathcal{D}}$ is a probability measure on $\mathcal{D}$. We assume that a set $\{m\in\mathcal{D}\mid
\tilde{p}_{k,l}^m\neq0\}$ is finite for all $k,l\in\mathcal{D}$ and define a product $\circ$ on the free vector space $\mathbb{C}H_\mathcal{M}$ with a basis $H_\mathcal{M}=\{z_i\}_{i\in\mathcal{D}}$ by
\[
z_i\circ
z_j=\sum_{k\in\mathcal{D}}\tilde{p}_{i,j}^kz_k
\]
for each $i,j\in\mathcal{D}$. If a triple $(H_\mathcal{M},\circ,*)$ forms a hypergroup for some involution $*$ on $\mathbb{C}H_\mathcal{M}$ then we say that the OQRW $\mathcal{M}$ on $\mathcal{D}$ with the initial state $\rho^{(0)}$ produces a hypergroup $H_\mathcal{M}$.
\begin{rema}
As was the case with Wildberger's construction, the above product on $\mathbb{C}H_\mathcal{M}$ does not always satisfy the associativity. Also, a produced hypergroup by an OQRW with an initial state is not commutative and unital in general.
\end{rema}
\begin{theo}\label{HGrecov}
For any hypergroup $H=\{x_i\}_{i\in
I(H)}$, there is an OQRW $\mathcal{M}$ on a distance set with an initial state $\rho^{(0)}$ producing a hypergroup $H_\mathcal{M}$ whose structure constants coincide with the ones of $H$, and hence $H_\mathcal{M}$ is isomorphic to $H$. 
\end{theo}
\begin{proof}
Let $\{q_{i,j}^k\}_{i,j,k\in
I(H)}$ be the structure constants of $H$. Suppose $\mathcal{D}$ is the distance set $I(H)$ associated with $H$ and $\mathcal{K}$ is a separable Hilbert space with an orthonormal basis $\{\ket{n}\}_{n\in\mathcal{D}}$. Take an arbitrary isometry operator $U_{i,j;k}$ on a Hilbert space $\mathcal{H}$ for $i,j,k\in\mathcal{D}$ and an arbitrary operator $\rho_0^{(0)}\in\mathcal{B}_1(\mathcal{H})_{+,1}$.

We define an operator $B_{i,j;k}\in\mathcal{B}(\mathcal{H})$ for each $i,j,k\in\mathcal{D}$ and an initial state $\rho^{(0)}\in\mathcal{B}_1(\mathcal{H}\otimes\mathcal{K})_{+,1}$ by
\[
B_{i,j;k}=\sqrt{q_{k,j}^i}U_{i,j;k},\quad
\rho^{(0)}=\rho_0^{(0)}\otimes\ket{0}\!\bra{0}.
\]
Then it is clear that the family $\{B_{i,j;k}\}_{i,j,k\in\mathcal{D}}$ satisfies \eqref{OQRWB}. Let $\mathcal{M}_H=\{\mathcal{M}_k\}_{k\in\mathcal{D}}$ be the OQRW on $\mathcal{D}$ associated with the family $\{B_{i,j;k}\}_{i,j,k\in\mathcal{D}}$. Since $x_0$ is the unit of the algebra $\mathbb{C}H$, we have
\[
\rho^{(1;(l))}=\mathcal{M}_l(\rho^{(0)})=\sum_{i\in\mathcal{D}}B_{i,0;l}\rho_0^{(0)}B_{i,0;l}^*\otimes\ket{i}\!\bra{i}=U_{l,0;l}\rho_0^{(0)}U_{l,0;l}^*\otimes\ket{l}\!\bra{l},
\]
and hence
\begin{align*}
&\rho^{(2;(l,k))}=\mathcal{M}_k\circ\mathcal{M}_l(\rho^{(0)})=\sum_{i\in\mathcal{D}}\sum_{j\in\mathcal{D}}B_{i,j;k}\rho_j^{(1;(l))}B_{i,j;k}^*\otimes\ket{i}\!\bra{i}\\
&=\sum_{i\in\mathcal{D}}B_{i,l;k}\rho_l^{(1;(l))}B_{i,l;k}^*\otimes\ket{i}\!\bra{i}=\sum_{i\in\mathcal{D}}q_{k,l}^iU_{i,l;k}U_{l,0;l}\rho_0^{(0)}U_{l,0;l}^*U_{i,l;k}^*\otimes\ket{i}\!\bra{i}.
\end{align*}
Thus, we have $\tilde{p}_{k,l}^m={\rm
Tr}(\rho_m^{2;(l,k)})=q_{k,l}^m$ for all $l,k,m\in\mathcal{D}$.

By an involution $*$ on $\mathbb{C}H_\mathcal{M}$ induced from the one on $\mathbb{C}H$ via a bijection $:z_i\mapsto
x_i$, the triple $(\mathbb{C}H_\mathcal{M},\circ,*)$ forms a hypergroup which is isomorphic to the hypergroup $H$.
\end{proof}
\begin{rema}
For a hypergroup $H$, the OQRW $\mathcal{M}_H$ on $I(H)$ defined in the proof of the previous theorem is not unique and depends on the choice of isometry operators $U_{i,j;k}$'s, however we simply call $\mathcal{M}_H$ the OQRW on $I(H)$ derived from $H$. Since the family $\{B_{i,j;k}\}_{i,j,k\in
I(H)}$ satisfies the condition $(2)$ in Proposition \ref{inde}, the family $\mathcal{M}_H$ is linear independent on $\mathcal{L}(\mathcal{B}_{1}(\mathcal{H}))$.

We can also obtain an OQRW on a distance set from a given pre-hypergroup by the same way, and it satisfies the condition $(2)$ in Proposition \ref{inde}.

\end{rema}

Let $(\Gamma,v_0)$ be a pointed graph. In the following two examples, we shall construct an OQRW on the distance set $I(\Gamma,v_0)$ under the idea that a quantum walker jumps between vertices in $\Gamma$ from $v_0$ and a measured distance is one between a base point and a vertex at which the walker arrives. Here, we shall need to consider the structure of the graph. Then it is a little natural to put as
\begin{equation}\label{gassu}
B_{i,0;k}=\begin{cases}
{\bf
1}\quad(i=k)\\
O\quad(i\neq
k)
\end{cases}
\end{equation}
for all $i,k\in
I(\Gamma,v_0)$, as a special case of the condition $(2)$ in Proposition \ref{inde}.
\begin{exam}\label{ex1}
Suppose a quantum walker jumps between vertices in the Cayley graph $\mathcal{C}_4$ in Example \ref{cycle4} from $\overline{0}$. Let $\mathcal{D}=I(\mathcal{C}_4,\overline{0})=\{0,1,2\}$, $\mathcal{H}=\mathbb{C}^2$ and $\mathcal{K}=\mathbb{C}^3$. We define an OQRW on $\mathcal{D}$ by
\[
B_{i,j;k}=\begin{cases}
B\quad((i,j,k)=(0,1,1))\\
C\quad((i,j,k)=(2,1,1))\\
{\bf
1}\quad((i,j,k)=(1,2,1),(1,1,2),(0,2,2))\\
{\bf
1}\quad(j=0,i=k\mbox{ or }k=0,i=j)\\
O\quad(\mbox{otherwise})
\end{cases}
\]
for some matrices $B,C$ acting on $\mathcal{H}$ with $B^*B+C^*C={\bf
1}$.

For an initial state $\rho^{(0)}=\rho_0^{(0)}\otimes\ket{0}\!\bra{0}$, we have
\begin{align*}
\rho^{(2;(k,l))}=\begin{cases}
B\rho_0^{(0)}B^*\otimes\ket{0}\!\bra{0}+C\rho_0^{(0)}C^*\otimes\ket{2}\!\bra{2}\quad((k,l)=(1,1)),\\
\rho_0^{(0)}\otimes\ket{0}\!\bra{0}\quad((k,l)=(0,0),(2,2)),\\
\rho_0^{(0)}\otimes\ket{1}\!\bra{1}\quad((k,l)=(0,1),(1,0),(1,2),(2,1)),\\
\rho_0^{(0)}\otimes\ket{2}\!\bra{2}\quad((k,l)=(0,2),(2,0)).
\end{cases}
\end{align*}
We put
\[
B=\frac{1}{\sqrt{3}}\begin{pmatrix}1&1\\0&1\end{pmatrix},\quad
C=\frac{1}{\sqrt{3}}\begin{pmatrix}1&0\\-1&1\end{pmatrix},\quad\rho_0^{(0)}=\begin{pmatrix}x&0\\0&1-x\end{pmatrix}
\]
for some $0\leq
x\leq1$. For example, we have
\begin{align*}
\tilde{p}_{1,1}^0={\rm
Tr}(B\rho_0^{(0)}B^*)=\frac{1}{3}(2-x),\quad\tilde{p}_{1,1}^2={\rm
Tr}(C\rho_0^{(0)}C^*)=\frac{1}{3}(1+x).
\end{align*}
the numbers $\{\tilde{p}_{i,j}^k\}_{i,j,k\in\mathcal{D}}$ are the structure constants of the hermitian hypergroup $H_\mathcal{M}=\{z_0,z_1,z_2\}$ with the following structure:
\begin{align*}
z_1\circ
z_1=\frac{1}{3}(2-x)z_0+\frac{1}{3}(1+x)z_2;\quad
z_1\circ
z_2=z_1;\quad
z_2\circ
z_2=z_0.
\end{align*}
\end{exam}
\begin{exam}\label{ex2}
Suppose a quantum walker jumps between vertices in the Cayley graph ${\rm
Cay}(\mathbb{Z},\{\pm1\})$ in Example \ref{lattice} from $0$. Put $\mathcal{D}=I({\rm
Cay}(\mathbb{Z},\{\pm1\}),0)=\mathbb{N}_0$. Let $\mathcal{H}$ be a separable Hilbert space and $\mathcal{K}$ a Hilbert space with an orthonormal basis $\{\ket{n}\}_{n=0}^\infty$. We define an OQRW on $\mathcal{D}$ by
\[
B_{i,j;k}=\begin{cases}
B_k\quad(i=j+k,j\neq0,k\neq0)\\
C_k\quad(i=|j-k|,j\neq0,k\neq0)\\
{\bf
1}\quad(i=k,j=0\mbox{ or }i=j,k=0)\\
O\quad(\mbox{otherwise})
\end{cases}
\]
for some family $\{B_k,C_k\}_{k=1}^\infty$ of bounded operators on $\mathcal{H}$ satisfying $B_k^*B_k+C_k^*C_k={\bf
1}$ for all $k\in\mathbb{N}$.

For an initial state $\rho^{(0)}=\rho_0^{(0)}\otimes\ket{0}\!\bra{0}$, we have
\begin{align*}
\rho^{(2;(k,l))}=\begin{cases}
B_l\rho_0^{(0)}B_l^*\otimes\ket{k+l}\!\bra{k+l}+C_l\rho_0^{(0)}C_l^*\otimes\ket{|k-l|}\!\bra{|k-l|}\quad(k\neq0),\\
\rho_0^{(0)}\otimes\ket{l}\!\bra{l}\quad(k=0),
\end{cases}
\end{align*}
and hence 
\[
\tilde{p}_{k,l}^m=\begin{cases}
    {\rm
    Tr}(B_l\rho_0^{(0)}B_l^*)\quad(m=k+l,k\neq0),\\
    {\rm
    Tr}(C_l\rho_0^{(0)}C_l^*)\quad(m=|k-l|,k\neq0),\\
    1\quad(m=l,k=0\mbox{ or }m=k,l=0),\\
    0\quad(\mbox{othewise}).
  \end{cases}
\]
The OQRW $\mathcal{M}$ on $\mathbb{N}_0$ with the initial state $\rho^{(0)}$ produces a hypergroup if and only if
\[
{\rm
Tr}(B_i\rho_0^{(0)}B_i^*)={\rm
Tr}(C_i\rho_0^{(0)}C_i^*)=\frac{1}{2}
\]
for all $i\in\mathbb{N}$. Indeed, we can check that if there exist non-negative numbers $\{p_k,q_k\}_{k=1}^\infty$ with $p_k+q_k=1$ for all $k\in\mathbb{N}$ and a hypergroup $H=\{x_i\}_{i=0}^\infty$ has a product structure
\[
x_i\circ
x_j=p_jx_{|i-j|}+q_jx_{i+j}
\]
for each $i,j\in\mathbb{N}$, then $p_k=q_k=\frac{1}{2}$ holds for all $k\in\mathbb{N}$.
\end{exam}


\section{Hypergroup structures of OQRWs}\label{Hypergroup structure of OQRW}

Recall that each distance distribution of a random walk on a pointed graph $(\Gamma,v_0)$ equipped with the symmetry condition $(S)$ in Theorem \ref{ems}, can be described by the structure of the pre-hypergroup $H(\Gamma,v_0)$. In addition, if $(\Gamma,v_0)$ produces a hypergroup $H(\Gamma,v_0)$ then each product of transition matrices associated with $H(\Gamma,v_0)$ can be decomposable as
\begin{equation}\label{maincoro'}
P_{k_n}\cdots
P_{k_1}=\sum_{m\in
I(\Gamma,v_0)}q_{k_n,\ldots,k_1}^mP_m,
\end{equation}
where $q_{k_n,\ldots,k_1}^m$'s are the multi-structure constants of $H(\Gamma,v_0)$.

In this section, for an OQRW $\mathcal{M}$ on a distance set, we will show that an analogous result of \eqref{maincoro'} holds if and only if $\mathcal{M}$ has a property like (pre-)hypergroups. As a consequence, we can also compute a distribution of such an OQRW $\mathcal{M}$ on a distance set by a (non-associative) algebraic structure.

Let $\mathcal{D}$ be a distance set and $\{Q_{i,j}^k\}_{i,j,k\in\mathcal{D}}$ a family of numbers $Q_{i,j}^k\geq0$ such that $Q_{i,j}^k=0$ for all but finitely many $k\in\mathcal{D}$ and $\sum_{k\in\mathcal{D}}Q_{i,j}^k=1$ for each $i,j\in\mathcal{D}$. We define a product on the free vector space $\mathbb{C}A$ of a family $A=\{Z_i\}_{i\in\mathcal{D}}$ of dummy symbols $Z_i$, by
\[
Z_i\circ
Z_j=\sum_{k\in\mathcal{D}}Q_{i,j}^kZ_k
\]
for each $i,j\in\mathcal{D}$. In general, this product is not always associative, commutative and unital. We also call the family $\{Q_{i,j}^k\}_{i,j,k\in\mathcal{D}}$ the structure constants of the (non-associative) algebra $\mathbb{C}A$. For each $k_1,\ldots,k_n\in\mathcal{D}$, a distribution $(Q_{k_1,\ldots,k_n}^m)_{m\in\mathcal{D}}$ on $\mathcal{D}$ is defined by
\[
(((Z_{k_1}\circ
Z_{k_2})\circ
Z_{k_3})\circ\cdots)\circ
Z_{k_{n-1}})\circ
Z_{k_n}=\sum_{m\in\mathcal{D}}Q_{k_1,\ldots,k_n}^mZ_m,
\]
as the multi-structure constants of a pre-hypergroup.

\begin{theo}\label{distM}
Let $\mathcal{M}=\{\mathcal{M}_k\}_{k\in\mathcal{D}}$ be an OQRW on a distance set $\mathcal{D}$ and $\{Q_{i,j}^k\}_{i,j,k\in\mathcal{D}}$ the above numbers. The family $\{B_{i,j;k}\}_{i,j,k\in\mathcal{D}}$ satisfies
\begin{equation}\label{HB}
\sum_{m\in\mathcal{D}}B_{m,j;l}^*B_{i,m;k}^*B_{i,m;k}B_{m,j;l}=\sum_{m\in\mathcal{D}}Q_{k,l}^mB_{i,j;m}^*B_{i,j;m}
\end{equation}
for all $i,j,m,k,l\in\mathcal{D}$ if and only if
\[
p^{(2;(k_1,k_2))}=d\left(\sum_{m\in
\mathcal{D}}Q_{k_2,k_1}^m\mathcal{M}_m(\rho)\right)
\]
for all $k_1,k_2\in\mathcal{D}$ and a state $\rho\in\mathcal{B}_{1}(\mathcal{H}\otimes\mathcal{K})_{+,1}$. Then we have
\[
p^{(n;{\bf
k})}=d\left(\sum_{m\in
\mathcal{D}}Q_{k_n,\cdots,k_1}^m\mathcal{M}_m(\rho)\right)
\]
for any tuple ${\bf
k}=(k_1,\ldots,k_n)\in\mathcal{D}^n$ and a state $\rho\in\mathcal{B}_{1}(\mathcal{H}\otimes\mathcal{K})_{+,1}$.
\end{theo}

\begin{proof}
First, since $B_{i,j;k}^*B_{i,j;k}\leq{\bf
1}$ for all $i,j,k\in\mathcal{D}$, the series of the left hand side of \eqref{HB} converges to a positive bounded operator in the $\sigma$-strong operator topology.

For $k_1,k_2\in\mathcal{D}$ and $\rho\in\mathcal{B}_{1}(\mathcal{H}\otimes\mathcal{K})_{+,1}$, we have
\begin{align}
&\label{Mk2}\rho^{(2;(k_1,k_2))}=\mathcal{M}_{k_2}\circ\mathcal{M}_{k_1}(\rho)\\
&=\sum_{i\in
\mathcal{D}}\sum_{j,m\in
\mathcal{D}}B_{i,j;k_2}B_{j,m;k_1}\rho_mB_{j,m;k_1}^*B_{i,j;k_2}^*\otimes\ket{i}\!\bra{i}.\nonumber
\end{align}
where the operators $\rho_j\in\mathcal{B}_1(\mathcal{H})$ is in Proposition \ref{oqwlemm}. Note that the choices of $\rho_j$'s depend on only $\rho$. Fix $i\in
\mathcal{D}$. We have
\begin{align}
&\label{dynd}{\rm
Tr}\left(\sum_{j,m\in
\mathcal{D}}B_{i,j;k_2}B_{j,m;k_1}\rho_mB_{j,m;k_1}^*B_{i,j;k_2}^*\right)\\
&=\sum_{m\in
\mathcal{D}}{\rm
Tr}\left(\rho_m\left(\sum_{j\in\mathcal{D}}B_{j,m;k_1}^*B_{i,j;k_2}^*B_{i,j;k_2}B_{j,m;k_1}\right)\right),\nonumber\\
&\label{algd}\sum_{m\in
\mathcal{D}}\sum_{j\in
\mathcal{D}}Q_{k_2,k_1}^m{\rm
Tr}(B_{i,j;m}\rho_jB_{i,j;m}^*)\\
&=\sum_{j\in
\mathcal{D}}\sum_{m\in
\mathcal{D}}Q_{k_2,k_1}^j{\rm
Tr}(B_{i,m;j}\rho_mB_{i,m;j}^*)\nonumber\\
&=\sum_{m\in
\mathcal{D}}{\rm
Tr}\left(\rho_m\sum_{j\in
\mathcal{D}}Q_{k_2,k_1}^jB_{i,m;j}^*B_{i,m;j}\right).\nonumber
\end{align}
Thus, if the family $\{B_{i,j;k}\}_{i,j,k\in\mathcal{D}}$ satisfies \eqref{HB} then the equations \eqref{dynd} and \eqref{algd} imply that
\begin{align*}
&{\rm
Tr}\left(\sum_{j,m\in
\mathcal{D}}B_{i,j;k_2}B_{j,m;k_1}\rho_mB_{j,m;k_1}^*B_{i,j;k_2}^*\right)=\sum_{m\in
\mathcal{D}}\sum_{j\in
\mathcal{D}}Q_{k_2,k_1}^m{\rm
Tr}(B_{i,j;m}\rho_jB_{i,j;m}^*),
\end{align*}
and hence the distribution $p^{(2;(k_1,k_2))}$ of $\mathcal{M}_{k_2}\circ\mathcal{M}_{k_1}(\rho)$ and the one of $\sum_{m\in
\mathcal{D}}Q_{k_2,k_1}^m\mathcal{M}_m(\rho)$ coincide by \eqref{Mk2}.

Conversely, suppose the distribution $p^{(2;(k_1,k_2))}$ of $\mathcal{M}_{k_2}\circ\mathcal{M}_{k_1}(\rho)$ is the one of $\sum_{m\in
\mathcal{D}}Q_{k_2,k_1}^m\mathcal{M}_m(\rho)$ for all $k_1,k_2\in\mathcal{D}$ and $\rho\in\mathcal{B}_{1}(\mathcal{H}\otimes\mathcal{K})_{+,1}$. For each $m\in\mathcal{D}$ and an arbitrary operator $\rho'\in\mathcal{B}_{1}(\mathcal{H})_{+,1}$, we put $\rho=\rho'\otimes\ket{m}\!\bra{m}$. Then we have
\[
{\rm
Tr}\left(\rho'\sum_{j\in
\mathcal{D}}B_{j,m;k_1}^*B_{i,j;k_2}^*B_{i,j;k_2}B_{j,m;k_1}\right)={\rm
Tr}\left(\rho'\sum_{j\in
\mathcal{D}}Q_{k_2,k_1}^jB_{i,m;j}^*B_{i,m;j}\right).
\]
By the correspondence between density matrices and normal states on the von Neumann algebra $\mathcal{B}(\mathcal{H})$ in \cite[Theorem 2.4.21]{brat-robi87}, the equation \eqref{HB} holds.

We shall show the second assertion by induction on $n$ under the condition \eqref{HB}. For $n=2$, we have shown it. Assume that the distribution $p^{(n;(k_1,\ldots,k_{n-1}))}$ of $\mathcal{M}_{k_{n-1}}\circ\cdots\circ\mathcal{M}_{k_1}(\rho)$ coincides with the one of $\sum_{m\in
\mathcal{D}}q_{k_{n-1},\cdots,k_1}^m\mathcal{M}_m(\rho)$. By the induction hypothesis, the $i$-coefficient of the distribution $p^{(n;{\bf
k})}$ of
\[
\mathcal{M}_{k_n}\circ\cdots\circ\mathcal{M}_{k_2}(\mathcal{M}_{k_1}(\rho))=\mathcal{M}_{k_n}\circ\cdots\circ\mathcal{M}_{k_2}(\rho^{(1;(k_1))})
\]
coincides with
\begin{align*}
&\sum_{m\in\mathcal{D}}Q_{k_n,\ldots,k_2}^m{\rm
Tr}\left(\sum_{j\in\mathcal{D}}B_{i,j;m}\rho_j^{(1;(k_1))}B_{i,j;m}^*\right)\\
&=\sum_{l,m\in\mathcal{D}}Q_{k_n,\ldots,k_2}^m{\rm
Tr}\left(\rho_l\sum_{j\in\mathcal{D}}B_{j,l;k_1}^*B_{i,j;m}^*B_{i,j;m}B_{j,l;k_1}\right)\\
&=\sum_{l,m\in\mathcal{D}}Q_{k_n,\ldots,k_2}^m{\rm
Tr}\left(\rho_l\sum_{j\in\mathcal{D}}Q_{m,k_1}^jB_{i,l;j}^*B_{i,l;j}\right)\\
&=\sum_{j,l\in\mathcal{D}}\sum_{m\in\mathcal{D}}Q_{k_n,\ldots,k_2}^mQ_{m,k_1}^j{\rm
Tr}\left(B_{i,l;j}\rho_lB_{i,l;j}^*\right).
\end{align*}
We have $Q_{k_n,\ldots,k_1}^j=\sum_{m\in\mathcal{D}}Q_{k_n,\ldots,k_2}^mQ_{m,k_1}^j$ by the definition of $Q_{k_1,\ldots,k_n}^m$'s, and hence the above equals to
\begin{align*}
\sum_{j,l\in\mathcal{D}}Q_{k_n,\ldots,k_1}^j{\rm
Tr}\left(B_{i,l;j}\rho_lB_{i,l;j}^*\right)=\sum_{m\in\mathcal{D}}\sum_{j\in\mathcal{D}}Q_{k_n,\ldots,k_1}^m{\rm
Tr}\left(B_{i,j;m}\rho_lB_{i,j;m}^*\right)
\end{align*}
which is the $i$-coefficient of the distribution of $\sum_{m\in
\mathcal{D}}Q_{k_n,\cdots,k_1}^m\mathcal{M}_m(\rho)$. By induction, we have shown the second assertion.
\end{proof}

\begin{rema}
If an OQRW $\mathcal{M}$ on a distance set $\mathcal{D}$ satisfies \eqref{HB} for some $\{Q_{i,j}^k\}_{i,j,k\in\mathcal{D}}$ and \eqref{gassu} for all $i,k\in\mathcal{D}$ then each $B_{i,j,k}$ is a scalar multiplication of an isometry operator. 
\end{rema}

Let $\{Q_{i,j}^k\}_{i,j,k\in\mathcal{D}}$ be structure constants of an arbitrary (non-associative) algebra $\mathbb{C}A$. As the OQRW on $I(H)$ derived from a hypergroup $H$, we define an OQRW on $\mathcal{D}$ by
\[
B_{i,j;k}=\sqrt{Q_{k,j}^i}U_{i,j;k}
\]
for an arbitrary isometry operator $U_{i,j;k}$. Then the equation \eqref{HB} means the associativity of $\mathbb{C}A$, and hence we have the following corollary:


\begin{coro}\label{distH}
Let $\mathcal{M}_H=\{\mathcal{M}_k\}_{k\in
I(H)}$ be the OQRW on $I(H)$ derived from a hypergroup $H$. For any state $\rho\in\mathcal{B}_1(\mathcal{H}\otimes\mathcal{K})_{+,1}$ and tuple ${\bf
k}=(k_1,\ldots,k_n)\in
I(H)^n$ with $n\geq2$, we have
\[
p^{(n;{\bf
k})}=d\left(\sum_{m\in
I(H)}q_{k_n,\cdots,k_1}^m\mathcal{M}_m(\rho)\right),
\]
where $q_{k_n,\cdots,k_1}^m$ is the multi-structure constant of $H$.
\end{coro}

As the special case of the previous corollary, we have the following corollary:

\begin{coro}\label{QWP}
Suppose a pointed graph $(\Gamma,v_0)$ produces a hermitian hypergroup $H(\Gamma,v_0)$. Let $\{P_k\}_{k\in
I(\Gamma,v_0)}$ be the probability transition matrices associated with $H(\Gamma,v_0)$, and $\mathcal{M}_{H(\Gamma,v_0)}$ the OQRW on the distance set $I(\Gamma,v_0)$ derived from $H(\Gamma,v_0)$.

For a tuple ${\bf
k}=(k_1,\ldots,k_n)\in
I(\Gamma,v_0)^n$ with $n\geq2$ and an initial state $\rho^{(0)}=\sum_{i\in
I(\Gamma,v_0)}\rho_i^{(0)}\otimes\ket{i}\!\bra{i}\in\mathcal{B}_1(\mathcal{H}\otimes\mathcal{K})_{+,1}$, we have
\[
p^{(n;{\bf
k})}=(P_{k_n}\cdots
P_{k_1})p^{(0)},
\]
where we regard the raw vector $(P_{k_n}\cdots
P_{k_1})p^{(0)}$ as a column vector.
\end{coro}
\begin{proof}
Let $q_{k_n,\ldots,k_1}^m$'s be the multi-structure constants of $H(\Gamma,v_0)$. By the equation \eqref{matrixproduct}, the $i$-coefficient of the (column) vector $(P_{k_n}\cdots
P_{k_1})p^{(0)}$ coincides with
\[
\sum_{j\in
I(\Gamma,v_0)}\sum_{m\in
I(\Gamma,v_0)}q_{k_n,\ldots,k_1}^mp_{m,j}^i{\rm
Tr}(\rho_j)
\]
which is equal to the $i$-coefficient of the distribution of
\[
\sum_{m\in
I(\Gamma,v_0)}q_{k_n,\cdots,k_1}^m\mathcal{M}_m(\rho).
\]
By Corollary \ref{distH}, this completes the proof.
\end{proof}

Corollary \ref{QWP} is an analogy of \cite[Proposition 6.1]{apss12}, in OQRWs on distance sets. If $(\Gamma,v_0)$ satisfies the condition $(S)$ in Theorem \ref{ems} and produces a hermitian hypergroup then we have $p_{k_n,\ldots,k_1}^m=q_{k_n,\ldots,k_1}^m$, and hence the proof of Corollary \ref{QWP} implies that the $i$-coefficient of $p^{(n;{\bf
k})}$ coincides with
\[
\sum_{j,m\in
I(\Gamma,v_0)}p_{k_n,\ldots,k_1}^mp_j^{(0)}p_{m,i}^i
\]
for an initial state $\rho^{(0)}=\sum_{i\in\mathcal{D}}\rho_i^{(0)}\otimes\ket{i}\!\bra{i}\in\mathcal{B}_1(\mathcal{H}\otimes\mathcal{K})_{+,1}$. In particular, if we put $\rho^{(0)}=\rho_0^{(0)}\otimes\ket{0}\!\bra{0}$ for some $\rho_0^{(0)}\in\mathcal{B}_1(\mathcal{H})_{+,1}$ then we have
\[
p^{(n;{\bf
k})}=(p_{k_n,\ldots,k_1}^i)_{i\in
I(\Gamma,v_0)}.
\]

In the following two examples, we shall give examples of an OQRW on a distance set, whose $B_{i,j;k}$'s are not scalar multiplications of some isometry operators and satisfy the condition \eqref{HB}

\begin{exam}\label{ex3}
Let $\mathcal{D}$ be a distance set and $\mathcal{H}$ be a separable Hilbert space. Suppose $\{A_i\}_{i\in\mathcal{D}}$ is a family of bounded operators $A_i$ on $\mathcal{H}$ satisfying
\[
\sum_{i\in\mathcal{D}}A_i^*A_i={\bf
1}
\]
in the strong operator topology and mutually commuting.

In the case when $\mathcal{D}=\{0,1,\ldots,N-1\}$ for some $N\in\mathbb{N}$ and $\dim\mathcal{H}=n$ for some $n\in\mathbb{N}$, let $A_i$'s be arbitrary commuting square matrices of order $n$ and $\{\lambda_{i,j}\}_{j=1}^n$ eigenvalues of $A_i$ such that $\sum_{i=0}^{N-1}|\lambda_{i,j}|^2=1$ for all $j=1,\ldots,n$. Then $\{A_i\}_{i=0}^{N-1}$ satisfies $\sum_{i=0}^{N-1}A_i^*A_i={\bf
1}$ by the simultaneous diagonalization. We can choose $A_i$'s as matrices which are not scalar multiplications of unitary matrices. For example, when $n=2,N=1$ the matrices
\begin{align}
&\label{exma}A_0=\frac{1}{2\sqrt{6}}\begin{pmatrix}\sqrt{3}+\sqrt{2}&\sqrt{3}-\sqrt{2}\\\sqrt{3}-\sqrt{2}&\sqrt{3}+\sqrt{2}\end{pmatrix},\\
&\label{exmb}A_1=\frac{1}{2\sqrt{6}}\begin{pmatrix}2+\sqrt{3}&-2+\sqrt{3}\\-2+\sqrt{3}&2+\sqrt{3}\end{pmatrix}.
\end{align}
commute and satisfy $A_0^*A_0+A_1^*A_1={\bf
1}$.

Let $\{Q_{i,j}^k\}_{i,j,k\in\mathcal{D}}$ be structure constants of an arbitrary (non-associative) algebra $\mathbb{C}A$. An OQRW $\mathcal{M}=\{\mathcal{M}_k\}_{k\in\mathcal{D}}$ on $\mathcal{D}$ defined by
\[
B_{i,j;k}=A_i
\]
for each $i,j,k\in\mathcal{D}$ satisfies \eqref{distM}. However, all $\mathcal{M}_k$'s coincide, and hence the family $\{\mathcal{M}_k\}_{k\in\mathcal{D}}$ is not linear independent. 

For an initial state $\rho^{(0)}=\sum_{i\in\mathcal{D}}\rho_i^{(0)}\otimes\ket{i}\!\bra{i}\in\mathcal{B}_1(\mathcal{H}\otimes\mathcal{K})_{+,1}$ and $k\in\mathcal{D}$, put $\rho^{(1)}=\mathcal{M}_k(\rho^{(0)})$. Theorem \ref{distM} implies that
\[
p^{(n;{\bf
k})}=d(\mathcal{M}_{k_n}\circ\cdots\circ\mathcal{M}_{k_1}(\rho^{(0)}))=d(\mathcal{M}_{k}(\rho^{(0)}))=p^{(1)}
\]
for all ${\bf
k}=(k_1,\ldots,k_n)\in\mathcal{D}^n$. This means that the distribution $p^{(1)}$ is stationary.
\end{exam}

\begin{exam}
Let $N=1$ and $A_0,A_1$ the operators in Example \ref{ex3}. If we define OQRW $\mathcal{M}=\{\mathcal{M}_k\}_{k\in\mathcal{D}}$ on the distance set $\mathcal{D}=\{0,1\}$ by
\[
B_{0,j;0}=A_0,\quad
B_{1,j;0}=A_1,\quad
B_{i,j;1}=\frac{1}{\sqrt{2}}{\bf
1}
\]
for each $i,j\in\mathcal{D}$ then we can check that the family $\{B_{i,j;k}\}_{i,j,k\in\mathcal{D}}$ satisfies the condition \eqref{HB} with
\begin{align*}
&Q_{0,0}^0=Q_{0,1}^0=1,\quad
Q_{0,0}^1=Q_{0,1}^1=Q_{1,0}^0=Q_{1,1}^0=0,\quad
Q_{1,0}^1=Q_{1,1}^1=1.
\end{align*}
The above numbers $\{Q_{i,j}^k\}_{i,j,k\in\mathcal{D}}$ is the structures constants of the left zero semigroup $LO_2$ of order two.

For a state $\rho=\rho_0\otimes\ket{0}\!\bra{0}+\rho_1\otimes\ket{1}\!\bra{1}$, we have 
\begin{align*}
\mathcal{M}_k(\rho)=\begin{cases}
A_0(\rho_0+\rho_1)A_0^*\otimes\ket{0}\!\bra{0}+A_1(\rho_0+\rho_1)A_1^*\otimes\ket{1}\!\bra{1}\quad(k=0)\\
\frac{1}{2}(\rho_0+\rho_1)\otimes\ket{0}\!\bra{0}+\frac{1}{2}(\rho_0+\rho_1)\otimes\ket{1}\!\bra{1}\quad(k=1).
\end{cases}
\end{align*}
If $A_0$ and $A_1$ are the matrices defined as \eqref{exma} and \eqref{exmb} then the family $\{B_{i,j;k}\}_{i,j,k=0,1,2}$ satisfies the condition $(1)$ in Proposition \ref{inde}.\end{exam}

\subsection*{Acknowledgements}
The authors would like to express gratitude to Akito Suzuki, Hiromichi Ohno and Yasumichi Matsuzawa for their helpful comments. This work was supported by JSPS KAKENHI Grant-in-Aid for Research Activity Start-up (No. 19K23403).

\end{document}